\documentclass[letterpaper,12pt]{article}

\usepackage{abstract}
\usepackage{comment}
\usepackage{tikz}
\usepackage{graphicx} 
\usepackage{amsmath}
\usepackage{threeparttable}
\usepackage{authblk}

\usepackage[ruled,vlined, Algorithm]{algorithm}
\usepackage{algpseudocode}
\usepackage{bm}
\usepackage{amsfonts}
\usepackage{amsthm}
\usepackage{natbib}
\usepackage{setspace}
\usepackage{multicol}
\usepackage{xcolor}
\usepackage{graphics}

\newtheorem{theorem}{Theorem}[section]
\newtheorem{corollary}{Corollary}[theorem]

\theoremstyle{definition}
\newtheorem{remark}{Remark}[section]

\newcommand{\blind}{0}

\begin{document}

\def\spacingset#1{\renewcommand{\baselinestretch}%
{#1}\small\normalsize} \spacingset{1}

\if0\blind
{
  \title{\bf  Stochastic gradient descent-based inference for dynamic network models with attractors}
  \author{
    Hancong Pan\\
    Department of Mathematics and Statistics, Boston University\\
    and \\
    Xiaojing Zhu\\
    Department of Mathematics and Statistics, Boston University\\
    and \\
    Cantay Caliskan \\
    Goergen Institute for Data Science, University of Rochester\\
    and \\
    Dino P. Christenson \\
    Department of Political Science, Washington University in St. Louis\\
    and \\
    Konstantinos Spiliopoulos \\
    Department of Mathematics and Statistics, Boston University\\
    and \\
    Dylan Walker \\
    Argyros School of Business and Economics, Chapman University\\
    and \\
    Eric D. Kolaczyk \\
    Department of Mathematics and Statistics, McGill University\\
    E-mail: eric.kolaczyk@mcgill.ca
  }
  \maketitle
} \fi

\if1\blind
{
  \bigskip
  \bigskip
  \bigskip
  \begin{center}
    {\LARGE\bf Stochastic gradient descent based inference for dynamic network models with attractors}
\end{center}
  \medskip
} \fi
\bigskip
\begin{abstract}
In Coevolving Latent Space Networks with Attractors (CLSNA) models, nodes in a latent space represent social actors, and edges indicate their dynamic interactions. Attractors are added at the latent level to capture the notion of attractive and repulsive forces between nodes, borrowing from dynamical systems theory. However, CLSNA reliance on MCMC estimation makes scaling difficult, and the requirement for nodes to be present throughout the study period limit practical applications.  
We address these issues by (i) introducing a Stochastic gradient descent (SGD) parameter estimation method, (ii) developing a novel approach for uncertainty quantification using SGD, and (iii) extending the model to allow nodes to join and leave over time. Simulation results show that our extensions result in little loss of accuracy compared to MCMC, but can scale to much larger networks. We apply our approach to the longitudinal social networks of members of US Congress on the social media platform X. Accounting for node dynamics overcomes selection bias in the network and uncovers uniquely and increasingly repulsive forces within the Republican Party. Supplemental materials for the article are available online.
\end{abstract}

\noindent%
\vfill

\newpage
\spacingset{1.75} 

\date{}
\doublespacing

\newpage

\section{Introduction}
We consider the problem of modeling dynamic networks, a collection of network graphs $G_t$ indexed over times $t$. In particular, we focus on a specific class of temporal network models known as CLSNA models, first developed in \cite{10.1093/jrsssa/qnad008}. In a CLSNA model, the relations and interactions between nodes and certain attributes of the nodes influence each other as they co-evolve over time.

Latent space network models are a commonly used class of models for static networks \citep{hoff2002latent} where the probability of a relation between actors depends on the positions of individuals in an unobserved “social space” or latent space. Several dynamic latent space network models have been proposed in the literature, which  involve embedding nodes of temporal networks into a latent Euclidean space, thus allowing each actor to have a temporal trajectory in the latent space. Earlier approaches modeling the transition and evolution of actors dictate that latent positions
evolve over time in a Markov fashion \citep{sarkar2005dynamic,sewell2015analysis,sewell2015latent,sewell2016latent}. In the CLSNA model \citep{10.1093/jrsssa/qnad008}, temporal evolution of latent positions depends on network connectivity via the presence of attractors (a concept that is fundamental to dynamical systems) at the latent level. The CLSNA model has been shown to be effective at disentangling positive (attractive) and negative (repulsive) forces among
political elites and the public when applied to longitudinal social networks from the social media platforms X and Reddit \citep{10.1093/jrsssa/qnad008}.

The CLSNA model has adeptly captured the nuances of polarization in American politics, offering valuable insights into the past decade. Despite its efficacy, the model's potential could be further realized by advancing the inference methods to address the challenges of scalability. The challenges of scalability arise in two key aspects. The first aspect is an increasing number of nodes. The second is changes to sets of nodes over time.

The first aspect addresses the reliance on MCMC, as used in \cite{10.1093/jrsssa/qnad008}, for estimating model parameters and latent positions. This approach, even though it is quite accurate, becomes prohibitively computationally expensive when scaling to networks with more than a couple hundred nodes.
Various methods have been proposed to reduce computational costs for dynamic latent position network models. \cite{sarkar2005dynamic} introduced a two-stage procedure to optimize the likelihood: multidimensional scaling followed by conjugate gradient update rule. \cite{raftery2012fast} used the case-control likelihood approximation to speed up the estimation algorithm. 
In \cite{liu2021variational}, the authors  propose to use a variational inference algorithm for estimating the dynamic latent position of the network model.

For the second aspect, the original formulation of CLSNA models assumes that all nodes are present at all time points. In the real world of dynamic networks, the sets of nodes tend to change over time. \cite{sewell2015latent} develop a model for dynamic networks in which only a subset of edges are observed and the edges are missing at random. \cite{10.1093/jrsssa/qnad008} handles the issue by pre-processing the data and keeping a subset of congress members who are present during the whole period of time studied.

In this paper, we directly address the above two challenges of scalability. Specifically, we (i) develop a Stochastic Gradient Descent (SGD) parameter estimation method that allows us to scale to larger networks; (ii) accompany that by a novel corresponding approach to variance estimation that builds on the notion of Laplace approximation; and (iii) implement these advances within an extension of the CLSNA framework that allows nodes to join and leave the network. Motivated by Laplace's approximation we replace the true posterior probability distribution by a Gaussian, with the mean at the MAP solution and precision matrix determined by the observed Fisher information. \cite{rue2009approximate} developed the class of INLA algorithms to perform Laplace's approximation when the precision matrix is sparse, for instance in the case of the conditional independence network structure in state-space models and Gaussian Markov Random Fields (GMRFs). However, in the class of CLSNA models that we study here, the assumption of sparsity does not remain applicable, necessitating the novel development we offer here. 

We use the Congressional Hashtag Network from the platform formerly known as Twitter, now referred to as ``X", as detailed in \cite{10.1093/jrsssa/qnad008}. The network was constructed from tweets by US congress-persons from 2010 to 2020. Each year, a binary network was formed, wherein nodes represented sitting members of Congress. Edges between two nodes indicated that their common hashtag usage exceeded the annual average among all pairs of congresspersons. We will revisit the X network to highlight the importance of the two aspects of scalability we propose in this paper: a flexible extended model that accommodates varying sets of nodes, and a fast and scalable SGD-based model inference method. We note that the simplicity and effectiveness of the proposed SGD-based inference and uncertainty quantification developed in this paper is potentially a valuable methodology for generic statistical models. We plan to investigate this direction in future works.

The rest of the paper is organized as follows. Section \ref{S:ExtendedCLSNA} reviews the standard CLSNA model in \cite{10.1093/jrsssa/qnad008} and presents the extended CLSNA model that enables us to accommodate real world dynamic networks with nodes entering and leaving the networks. Section \ref{S:SGDinference} introduces the two-stage SGD-based inference algorithm: the SGD-based point estimation method and SGD-based variance estimation method. Section \ref{S:Simulation } provides results from simulations with factorial designs to assess the capability of the proposed algorithm under different settings.\footnote{Code for this paper can be found at https://github.com/KolaczykResearch/SGD4CLSNA} In Section \ref{S:TwitterDataAnalysis}, we report the results from fitting the extended model to the X Congressional Hashtag Network. 

\section{Extended CLSNA model}\label{S:ExtendedCLSNA}
The CLSNA model of \cite{10.1093/jrsssa/qnad008} assumes that all nodes are present at all time points. In the real world of dynamic networks, the sets of nodes tend to change over time. Take the X Congressional Hashtag Network as an example. Figure \ref{count} shows numbers of re-elected and newly elected Democratic and Republican congress members in the X network. We observe significant number of changes in the sets of nodes present at each time because congress members only serve a fixed-year term before they are considered for reelection and because newly elected or sitting congress members decided to open an X account.

In order to use the original CLSNA model, \cite{10.1093/jrsssa/qnad008} keep a subset of congress members who were in office and on X during the whole period of time. After applying this filtering criteria, the number of nodes was reduced from around 500 to 207.

By developing a model that accommodates nodes entering and leaving the network, we could make inference based on the full data set. As a result, we are able to obtain latent position estimates at each time and estimates of the attractive or repulsive forces between a given node and its neighbors (influencing its movement over time) that are unbiased by factors like time in office and early technology adoption.
\begin{figure}[h]
\centering
\includegraphics[width=0.5\textwidth]{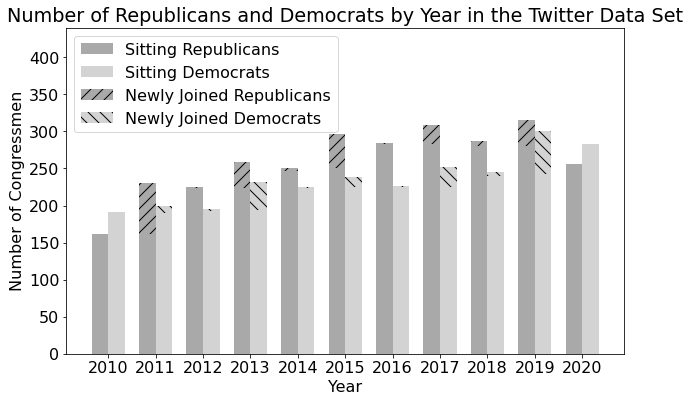}
\caption{Numbers of re-elected and newly elected Democratic and Republican congress members in the X hashtag network}
\label{count}
\end{figure}
\subsection{Model Definition}
Let $G_t=(V_t,E_t)$ be a network evolving in (discrete) time $t$, with vertex set $V_t$ and edge set $E_t$.  We allow for $V_t$ to vary over time.  Let $Y_t$ be the (random) adjacency matrix at time $t$ corresponding to $G_t$. 
Data is in the form of time series of adjacency matrices $\{y_t: t=1,\cdots, T\}$, where $y_{t, ij}=1$ if there is an edge between node $i$ and node $j$ at time $t$ and $y_{t, ij}=0$ otherwise. 

We model $Y_t$ using a latent space approach, with attractors added in the latent level to capture the notions of attractive and repulsive forces. In addition, we want to accommodate nodes entering and leaving the network. Let $ Z_{t,i} \in \mathbb{R}^p$ be a time-indexed latent (i.e., unobserved) position for node $i$ in $p$-dimensional Euclidean space, and let $ Z_t = \{Z_{t,i}\}$.  Assume that each of the $N_t$ nodes falls into one of two groups, i.e.,  Democratic and Republican, with known node label $\pi(i) \in \mathcal{C}$ for node $i$, where $\mathcal{C} = \{1, 2\}$ is the set of group labels. 

Formally, we define our model as follows:

\begin{equation}
     Y_{t,ij} | p_{t,ij} \sim Bernoulli(p_{t,ij})
     \label{eq:bernoulli}
\end{equation}

where at time $t=1$, 
\begin{equation}
    \hbox{logit}(p_{t,ij})  =  \alpha - s( z_{t,i},  z_{t,j}) 
    \label{eq:logit1}
\end{equation}
\begin{equation}
    Z_{t,i}\sim Normal  (0,\tau^2 I_p) \enskip .
\end{equation}

If, at time $t\geq 2$, if node $i$ was absent from the network at time $t-1$,
\begin{equation}
    \hbox{logit}(p_{t,ij})  =  \alpha + \delta Y_{t-1,ij} - s( z_{t,i},  z_{t,j}) 
    \label{eq:logit2}
\end{equation}
\begin{equation}
   % Z_{t,i}\sim Normal  (\mu_{t,i},\phi^2 I_p)
    Z_{t,i}\sim Normal  (\beta_{t,i} Z_{t_{0},i}+(1-\beta_{t,i})\mu_{t,i},\phi^2 I_p)
    \label{eq:new2}
\end{equation}
\begin{equation}
    \mu_{t,i} = \overline{z}^{\pi(i)}_{t-1} \enskip ,
\end{equation}
 where  $\beta_{t,i}\in[0,1]$ and $Z_{t_{0},i}$ denotes the  potential position  of the specific node $i$ at the most recent past time $t_{0}<t-1$ that it was present at the network. If node i has never been in the network before time $t$, then we simply set  $\beta_{t,i}=0$. See also Remark \ref{R:NodePresentBefore}.

If at time $t\geq 2$, node $i$ was present at the network  at time $t-1$, 
\begin{equation}
    \hbox{logit}(p_{t,ij})  =  \alpha - s( z_{t,i},  z_{t,j}) 
    \label{eq:logit3}
\end{equation}
\begin{equation}
  Z_{t,i} | Z_{t-1,i}=z_{t-1,i} \sim Normal  (\mu_{t,i},\sigma^2 I_p) 
  \label{eq:att}
\end{equation}
\begin{equation}
     \mu_{t,i} = {z}_{t-1, i}+\gamma^{\pi(i)}_{w} A_i^w(z_{t-1}, Y_{t-1})+\gamma_{b}A_i^b(z_{t-1}, Y_{t-1}) \enskip .
     \label{eq:z1}
\end{equation}

Here $s(\cdot, \cdot)$ is the Euclidean distance, $s(z_{t,i}, z_{t, j}) = ||z_{t,i}-z_{t, j}||_2$. $A_i^w$ and $A_i^b$ are the within group $\pi(i)$ (the group node $i$ belongs to) and between groups attractor functions for node $i$ in $Y_{t-1}$ respectively. In addition,  $\gamma^{\pi(i)}_{w}$ and $\gamma_{b}$ are the attraction parameters within group $\pi(i)$ and between groups attraction parameters respectively. The two attractors for node $i$ are defined as follows\footnote{This application employs a two-group CLSNA model with attractors that mimic attractive and repulsive force. It is a specific version of the general CLSNA model class. Each node is assigned to one of two labeled groups, and movements are dictated by attractors shaped by neighboring nodes' influence.}, 
\begin{equation}
A_i^w( z_{t-1},Y_{t-1}) = \bar{ z}^1_{t-1,i}-  z_{t-1,i}, \,
\bar{ z}^1_{t-1,i} = \frac{1}{|S_1^{(i)}|} \sum_{j\in S_1^{(i)}}  z_{t-1,j}
\label{eq:A1}
\end{equation}
\begin{equation}
A_i^b( z_{t-1},Y_{t-1}) = \bar{ z}^2_{t-1,i}-  z_{t-1,i}, \,
\bar{ z}^2_{t-1,i} = \frac{1}{|S_2^{(i)}|} \sum_{j\in S_2^{(i)}}  z_{t-1,j} \enskip ,
\label{eq:A2}
\end{equation}
which are the discrepancies of $ z_{t-1,i}$ from two local averages at time $t-1$. The sets $S_1^{(i)}$ and $S_2^{(i)}$ are defined based on group membership and network connectivity as follows:
\begin{enumerate}
\item 
$ S_1^{(i)} = \{ j \in \mathcal{N} \setminus i \,|\, Y_{ij} = 1, \pi(i) = \pi(j) \}$, neighbors of node $i$ in the same group
 \item 
 $  S_2^{(i)} = \{ j \in \mathcal{N} \setminus i \,|\, Y_{ij} = 1, \pi(i) \neq \pi(j) \}$, neighbors of node $i$ in a different group.
\end{enumerate}

Note that even though it is not explicitly mentioned in the notation, the sets $S_1^{(i)}$ and $S_2^{(i)}$ also depend on time because the connectivity sets are allowed to change with time.
When $S_1^{(i)}$ or $S_2^{(i)}$ or both are empty, we set $A_i^w( z_{t-1},Y_{t-1})=0$, or $A_i^b( z_{t-1},Y_{t-1})=0$ or both to be zero respectively. This implies that if node $i$ is not connected to any members of a specific group, it receives no force from that group.

\begin{remark}\label{R:NodePresentBefore}
    We remark that (\ref{eq:new2})  allows to cover the case where a given node $i$ leaves the network at a given time $t_{0}$ and re-enters again at a future time $t>t_{0}+1$. In that case,   (\ref{eq:new2}) shows that the mean of a new position can be defined to be  a convex combination of the last known position and the average of the position of the rest of the nodes within the group. If no further information exists, taking $\beta_{t,i}=\frac{1}{2}$ is a reasonable choice. Of course, if node $i$ enters the network for the first time at time $t$, then the intuitive choice is to take $\beta_{t,i}=0$.
\end{remark}

Note that we automatically recover the original CLSNA model when  $V_t\equiv V$ is fixed over time. In this case, at time $t\geq 2$, a node $i$ is always present at time $t-1$ and \eqref{eq:logit2} and \eqref{eq:new2} are redundant in that case.

In this proposed extended model, each node lies in a $p$-dimensional Euclidean latent space. The smaller the distance between two nodes in the latent space, the greater their probability of being connected, as in \eqref{eq:bernoulli}, \eqref{eq:logit1}, \eqref{eq:logit2}, and \eqref{eq:logit3}. The expressions in Eqs. (\ref{eq:A1}) and (\ref{eq:A2}) capture the discrepancy between the current latent position of node $i$ and the average of that of its current neighbors in groups $1$ and $2$, respectively. The corresponding parameters $\gamma^w_1, \gamma^w_2,$ and $\gamma^b$ represent attractive and repulsive forces respectively.
CLSNA model allows the network connectivity to enter the temporal evolution of latent positions in the form of attractors. Specifically, in the proposed model the evolution of latent positions for each node $i$ from $t-1$ to $t$ is modeled by the normal transition distribution in Eq. \eqref{eq:z1}, the mean vector of which depends not only on the latent position of itself at time $t-1$, but also on the two local averages, one from its neighbors in the same group, the other from its neighbors in a different group, as captured in \eqref{eq:A1} and \eqref{eq:A2}. Strength of attraction and repulsion toward local averages is summarized by the attractor functions and the associated parameters. The parameter $\delta$ captures edge persistence.  For $\delta>0$, the probability of an edge at time $t$ will be increased when one exists already at time $t-1$.

The extended model is a natural extension of the original model. If there are no nodes entering or leaving the networks, the extended model is identical to the original model of \cite{10.1093/jrsssa/qnad008}. When the sets of nodes change over time, the model classifies nodes into two categories and models the two types of nodes separately:
\begin{enumerate}
    \item For actors who are ``retained'' from time $t-1$ to $t$, the manner in which network connectivity influences the temporal evolution of latent positions remains consistent with the original model: the evolution of latent positions for the node $i$ from $t-1$ to $t$ is modeled by the normal transition distribution, the mean vector of which is a linear combination of the latent position of itself at time $t-1$ and the mean of both its neighbors in the same group and its neighbors in a different group.
    \item For an actor that is present at time $t$ but is absent at time $t-1$, we cannot make an educated guess of its current position at time $t$ because of its absence at time $t-1$. For this reason we choose to model the prior distribution of its latent position as a normal distribution, the mean vector of which is the mean of all members in the same group at $t-1$, which we denote by $\overline{z}^{\pi(i)}_{t-1}$.
\end{enumerate}

We also present a graphical representation of the proposed model in Figure \ref{graphical_rep}. The graphical representation is adapted from \cite{10.1093/jrsssa/qnad008}. This diagram illustrates the key dependencies in the model. The arrow from $Y_{t-1,ij}$ to $Y_{t,ij}$ represents the edge persistence effect described by Eq. \eqref{eq:logit2}, indicating that the presence of an edge at time $t-1$ influences the likelihood of an edge at time $t$. The arrow from $Y_{t-1,ij}$ to $Z_t$ reflects the influence of prior network connectivity on the evolution of latent positions, as captured by Eq. \eqref{eq:att}. The generation of edges based on latent position similarity, as described in Eq. \eqref{eq:bernoulli}, is depicted by the arrows from $Z_t$ to $Y_{t,ij}$. Finally, the arrow from $Z_{t-1}$ to $Z_t$ illustrates the influence of the previous position, i.e., at time $t-1$ on the position at time $t$.
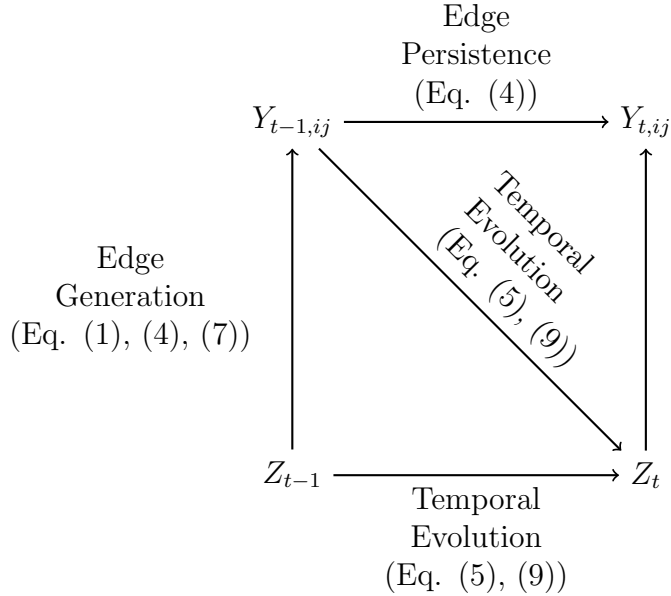
\begin{figure}[H]
\begin{center}
\begin{tikzpicture}[node distance = 4.5cm, thick]

        % Define nodes
        \node (1) {$Y_{t-1,ij}$};
        \node (2) [right of=1, xshift=0.2cm] {$Y_{t,ij}$};
        \node (3) [below of=1, yshift=-0.2cm] {$Z_{t-1}$};
        \node (4) [below of=2, yshift=-0.2cm] {$Z_{t}$};
        
        % Draw arrows with annotations
        \draw[->] (3) -- (1) node[midway, left, text width=4cm, align=center] {Edge\\ Generation\\ (Eq. \eqref{eq:bernoulli}, \eqref{eq:logit2}, \eqref{eq:logit3})};
        \draw[->] (4) -- (2) node[midway, right, text width=2.5cm, align=center] {};
        \draw[->] (3) -- (4) node[pos=0.5, sloped, below, text width=3.5cm, align=center] {Temporal\\ Evolution\\ (Eq. \eqref{eq:new2}, \eqref{eq:z1})};
        \begin{scope}
        \draw[->] (1) -- (2) node[midway, above, text width=2.5cm, align=center] {Edge\\ Persistence\\ (Eq. \eqref{eq:logit2})};
        \end{scope}
        \begin{scope}
        \draw[->] (1) -- (4) node[pos=0.65, sloped, above, text width=3.5cm] {Temporal\\ Evolution\\ (Eq. \eqref{eq:new2}, \eqref{eq:z1})};
        \end{scope}
\end{tikzpicture} 
\end{center}
\vspace{-1em}
\caption[Graphical model representation of the CLSNA model.]{\raggedright Graphical model representation of the CLSNA model. All key dependencies are shown, including edge generation, edge persistence and temporal evolution}
\label{graphical_rep}
\end{figure}

\section{SGD-based Inference Method} \label{S:SGDinference}
We develop a two-stage algorithm to make inferences based on the posterior $\pi(\theta, Z_{1:T}|Y_{1:T})$, where $\theta = (\alpha, \delta, \gamma^{w}, \gamma^{b})$. At the first stage we use SGD to compute a point estimate. At the second stage, we estimate marginal posterior standard deviations by a novel approach to quadratic approximation of the log-posterior density, which is also based on SGD.
\subsection{Point Estimate}\label{SS:PointEstimate}
The posterior distribution of latent positions and parameters is
\begin{align}
    & \pi( Z_{1:T},  \theta \,|\, Y_{1:T}) \label{Eq:PosteriorDistribution}\\
& \propto P(Y_{1:T},  Z_{1:T} \,|\,  \theta) \pi( \theta) \nonumber\\
& \propto 
\left( 
\prod_{t=2}^T P( Z_t \,|\,  Z_{t-1}, Y_{t-1}) P(Y_t \,|\,  Z_t, Y_{t-1}) 
\right)
P(Y_1 \,|\,  Z_1)P( Z_1) \pi( \theta) \nonumber\\
& = 
\left(
\prod_{t=2}^T 
\prod_{i=1}^{N_t} \left( P( Z_{t,i} \,|\,  Z_{t-1}, Y_{t-1}) \prod_{j: j \neq i} P(Y_{t,ij} \,|\,  Z_{t, i},  Z_{t, j}, Y_{t-1, ij}) \right)
\right) \cdot\nonumber\\
& \hspace{4in}
P(Y_1 \,|\,  Z_1)P( Z_1) \pi( \theta). \nonumber
\end{align}

The log posterior distribution (i.e., taking the log of \eqref{Eq:PosteriorDistribution}), therefore, can be written as sums of simpler functions. To optimize the posterior function with SGD, at each step, we randomly sample different terms from these sums. The key steps are:
\begin{enumerate}
\item Initialize $Z_{1:T}, \theta$.
\item Randomly sample indices of different terms from the different summands of the log-posterior distribution function. Compute the randomly sampled summand functions of the log-posterior distribution function.

\item Update $Z_{1:T}, \theta$ based on stochastic gradient descent using the output from step 2.
\end{enumerate}
Repeat steps 2-3 until convergence criterion is met.
In practice, we use gradient descent when the sample size is small and it is computationally feasible. In our implementation, we used gradient descent with momentum (when the sample size was small) or stochastic gradient descent with momentum (when the sample size was large), see  \cite{polyak1964some}.

\begin{algorithm}
\hspace*{\algorithmicindent} \textbf{Input}: Network time-series \(Y_{1:T}\). Current values \( \theta^{(k)}\), \( Z_{1:T}^{(k)}\) for parameters and latent positions. \\
\hspace*{\algorithmicindent} $\lambda$: step size for \(Z_{1:T}\) and for \(\theta\)\\
\hspace*{\algorithmicindent} \textbf{Output}: \(\theta^{(k+1)}, Z_{1:T}^{(k+1)}\)
\begin{algorithmic}[1]
\State Set \( \theta\), \( Z_{1:T}\) to be \( \theta^{(k)}\), \( Z_{1:T}^{(k)}\).
\State Uniformly sample the terms (factors) from the posterior distribution:
\begin{align}
\prod_{t=2}^T
\prod_{i=1}^N P( Z_{t,i} \,|\, Z_{t-1}, Y_{t-1}) \prod_{j: j \neq i} P(Y_{t,ij} \,|\, Z_{t, i}, Z_{t, j}, Y_{t-1, ij})P(Y_1 \,|\, Z_1)P( Z_1) \pi( \theta)
\end{align}
Denote the sampled terms by \(p_i\). Define \(\pi_{SGD}(\theta, Z_{1:T}|Y_{1:T})\) to be \(\prod_i p_i\)
\State Take a SGD step (see also Remark \ref{R:SignGradientStep})
\begin{align}
Z_{1:T}^{(k+1)} = &Z_{1:T}^{(k)}-\lambda \left(\frac{\partial \log \pi_{SGD}(\theta^{(k)}, Z_{1:T}^{(k)}|Y_{1:T})}{\partial Z_{1:T}}\right)\nonumber\\
\theta^{(k+1)} = &\theta^{(k)}-\lambda\left(\frac{\partial \log \pi_{SGD}(\theta^{(k)}, Z_{1:T}^{(k)}|Y_{1:T})}{\partial \theta}\right)\label{eq:sign}
\end{align}
\end{algorithmic}
\caption{\textsc{SGDstep}}
\end{algorithm}

Formally, Algorithm 1 is  the main building block for the SGD-based point estimation method, which we present as pseudo-code in Algorithm 2.

\begin{algorithm}
\hspace*{\algorithmicindent} \textbf{Input}: Network time-series \(Y_{1:T}\). Initial values \(\theta^{init}\), \( Z_{1:T}^{init}\) for parameters and latent positions. \\
\hspace*{\algorithmicindent} M: maximum number of iterations, \(\lambda\): step size for \(Z_{1:T}\) and for \(\theta\)\\
\hspace*{\algorithmicindent} \textbf{Output}: \(\theta^{(k)}, Z_{1:T}^{(k)}\)
\begin{algorithmic}[1]
\State Set initial values \(\theta^{(0)}\), \( Z_{1:T}^{(0)}\) to be \( \theta^{init}\), \( Z_{1:T}^{init}\).
\For{\(k \leftarrow 0\) to \(M\)}
\State Take a SGD step to optimize the posterior distribution: \((\theta^{(k+1)}, Z_{1:T}^{(k+1)})\) = \textsc{SGDstep}(\(\theta^{(k)}, Z_{1:T}^{(k)}\))
\State Stop if stopping criteria is met
\EndFor
\end{algorithmic}
\caption{\textsc{SGD-based Point estimation Method}}
\label{al1}
\end{algorithm}

\begin{remark}
In our work using gradient descent, we encountered a recurring challenge. For the majority of time points, the algorithm performs as expected, closely matching the true underlying values in simulation. However, there were a couple of time points at which an anomalous sign inversion occurs in one dimension of the latent position. In particular,  if $ z_{t,i} = [z_{t,i}^1, z_{t,i}^2, \ldots, z_{t,i}^p]$ was the true value, it was estimated to be $\tilde{z}_{t,i} = [-z_{t,i}^1, z_{t,i}^2, \ldots, z_{t,i}^p]$, even though all the estimates for other times match the truth. 

This problem is rooted in the complex landscape of our problem, where the model sometimes settles in local modes due to its non-convex nature.
The issue with local modes is similar to the issue of label-switching found in the  dynamic stochastic block model literature and the issue may not be solved without an extra assumption e.g. that most of the nodes do not
change group across two different time steps \citep{matias2017statistical}.

The resolution to this issue that we found to work well is to start with a CLSNA model with latent variables living in a higher dimensional space from the one we intend on having; that is, we aim to fit a model with $Z\in\mathbb{R}^{p}$ but we intentionally choose $Z\in\mathbb{R}^{q}$ with $q>p$ . Then, we use the first $p$ principal components of the vector of the latent position variables for $Z_{1:T}$. This method was the most effective of several alternatives attempted to remedy local modes. Intuitively, one can think of this approach as creating a path between modes thus allowing SGD to travel between modes in the higher dimension. In both the simulated and real data examples, we initially set \( p \) to 2 and chose \( q \) as 3. We begin by fitting a model with \( Z \) in \( \mathbb{R}^{3} \), followed by applying principal component analysis (PCA) to the resulting latent positions, thereby compressing \( Z \) into \( \mathbb{R}^{2} \) while maximizing information preservation. This reduced form is then used as the initial value for fitting a model with \( Z \) in \( \mathbb{R}^{2} \). This approach, albeit doubling the cost of fitting the model and performing PCA, provides more accurate and faster convergence towards the true values, as it effectively navigates through potential local optima that could hinder convergence in a single-stage fitting process.
\end{remark}

\begin{remark}\label{R:SignGradientStep}
        Motivated by the work in \cite{hinton2012neural}, at the initial stage of Algorithm 1 in Eq.\eqref{eq:sign} we use the sign of the gradient of the global parameters (+1/-1), i.e., we replace \eqref{eq:sign} by 
        \begin{align*}
        \theta^{(k+1)} = \theta^{(k)}-\lambda\hspace{0.1cm}\text{sign}\left(\frac{\partial \log \pi_{SGD}(\theta^{(k)}, Z_{1:T}^{(k)}|Y_{1:T})}{\partial \theta}\right).
        \end{align*}
        
        This is due to the fact that the magnitude and variance of the related gradient terms is much larger than the magnitude and variance of the gradients of the latent positions. This leads to more stability. As we approach the conclusion of the training phase, we switch to Eq.\eqref{eq:sign} for better accuracy.
 \end{remark} 
 
\subsection{Variance Estimation}\label{SS:VarianceEstimate}
The first stage with SGD above only gives point estimates, without quantifying uncertainty. We propose a novel approach to estimate the posterior variance for parameters of interest. It is a method for approximate Bayesian inference based on Laplace's method. Laplace's approximation, when applied in scenarios like state-space models and Gaussian Markov Random Fields, leverages the sparsity of the precision matrix as detailed by \cite{rue2009approximate}. However, in more general contexts, such as CLSNA, this assumption of sparsity is not always valid. In response to this, our proposed approach builds upon Laplace's foundational principles but adapts to the challenges of non-sparse structures. Specifically, it transforms the variance estimation problem into an optimization problem that can also be solved with SGD.

In this section, we propose a novel variance estimation algorithm.  
Our variance estimation method exploits the properties of the conditional distribution of a multivariate Gaussian distribution. Let \( x \) follow a multivariate normal distribution \( x \sim \mathcal{N}(\mu, \Sigma) \), where the block-wise mean and covariance matrices are defined as \(\mu = \left[\begin{smallmatrix} \mu_1 \\ \mu_2 \end{smallmatrix}\right]\) and \(\Sigma = \left[\begin{smallmatrix} \Sigma_{11} & \Sigma_{12} \\ \Sigma_{21} & \Sigma_{22} \end{smallmatrix}\right]\).
The conditional distribution of a subset vector \( x_1 \), given its complement vector \( x_2 \), is also a multivariate normal distribution \( x_1 | x_2 \sim \mathcal{N}(\mu_{1|2}, \Sigma_{1|2}) \). We will use \(\mu_{1|2}(x_2)\) to denote the conditional mean of \( x_1 \) given \( x_2 \).
With this notation, the variance and covariance can be respectively calculated as:
\begin{align}
    \text{Var}(x_2)&= \frac{1}{2}\frac{(\mu_2-\tilde{x}_2)^2}{\log p(x_1=\mu_1,x_2=\mu_2)-\log p(x_1=\mu_{1|2}(\tilde{x}_2),x_2=\tilde{x}_2)},
    \label{Eq:VarianceFormula1}\\
\Sigma_{12} &=  \frac{\Sigma_{22}}{ (\tilde{x}_2 - \mu_2)}(\mu_{1|2}(\tilde{x}_2) - \mu_1).
    \label{Eq:CovarianceFormula}
\end{align}
A more detailed discussion and proof of these properties can be found in the supplementary file. The results above for a generic normal distribution \( \mathcal{N}(\mu, \Sigma) \) are now adapted to estimate variance in the posterior distribution \( \pi(\cdot|Y) \), effectively applying the same theoretical concepts to the Bayesian analysis under a Laplace approximation.

\subsubsection{The Variance Estimation Algorithm}\label{SS:VarianceEstimationAlgorithm}
Eq.\ref{Eq:VarianceFormula1} and Eq.\ref{Eq:CovarianceFormula}  are the primary components in the second stage of our proposed SGD-based method for variance estimation of a single parameter, as detailed in Algorithm \ref{al2}. The first stage detailed in section \ref{SS:PointEstimate} computes a point estimate using SGD. In the second stage, with this point estimate, we redo the optimization with the constraint of fixing the parameter to be variance-estimated, adjusting it by a certain $\eta$ perturbation size. Thanks to Eq.\ref{Eq:VarianceFormula1} and Eq.\ref{Eq:CovarianceFormula}, we can then estimate the variance/covariance of this parameter by observing the changes in the posterior function and the changes in the other parameters resulting from the perturbation.

Algorithm \ref{al2} details the second stage of the SGD-based approach for estimating the variance of a single parameter (in here we choose $\alpha$).
\begin{algorithm}[H]
\caption{\textsc{SGD-based Var/Cov Estimation Method for a Single Parameter}}
\label{al2}
\hspace*{\algorithmicindent} \textbf{Input}: 
Network time-series \(Y_{1:T}\); mode of \(\log \pi(Z,\theta|Y)\), \(\theta^{*}\) (for all hyperparameters \(\alpha, \delta, \gamma^{w}, \gamma^{b}\)), \(Z_{1:T}^{*}\) (latent positions); \(M\): max iterations; \(\lambda\): step size for \(Z_{1:T}\) and for \(\theta\); \(\eta_{\alpha}\): perturbation size for \(\alpha\).

\hspace*{\algorithmicindent} \textbf{Output}: 
\(\widehat{\text{Var}(\alpha|Y)}\), \(\widehat{\text{Cov}(\alpha, \theta_i|Y)}\), \(\widehat{\text{Cov}(\alpha, Z_{1:T}|Y)}\)

\begin{algorithmic}[1]
    \State Set initial values \(\theta^{(0)}\), \(Z_{1:T}^{(0)}\) to \(\theta^{*}\), \(Z_{1:T}^{*}\).
    \For{\(k \leftarrow 0\) to \(M\)}
        \State Set \(\alpha^{(k)} = \alpha^{*} + \eta_{\alpha}\).
        \State \((\theta^{(k+1)}, Z_{1:T}^{(k+1)}) = \textsc{SGDstep}(\theta^{(k)}, Z_{1:T}^{(k)})\).
        \State Stop if criteria are met; let \(k\) be the corresponding iteration.
    \EndFor
    \State \(\widehat{\text{Var}(\alpha|Y)} = \frac{1}{2}\frac{\eta_{\alpha}^2}{\log \pi(\theta^{(0)},Z^{(0)}_{1:T}|Y_{1:T}) - \log\pi(\theta^{(k+1)},Z^{(k+1)}_{1:T}|Y_{1:T})}\)
    \State \(\widehat{\text{Cov}(\alpha, \theta_i|Y)} = \frac{\widehat{\text{Var}(\alpha|Y)}}{\eta_{\alpha}}(\theta_i^{(k+1)} - \theta_i^{(0)})\)
    \State \(\widehat{\text{Cov}(\alpha, Z_{1:T}|Y)} = \frac{\widehat{\text{Var}(\alpha|Y)}}{\eta_{\alpha}}(Z_{1:T}^{(k+1)} - Z_{1:T}^{(0)})\)
\end{algorithmic}
\end{algorithm}

In Figure \ref{fig:var_estimation} we aim to visualize this approach for a a Bivariate Distribution.
\begin{figure}[htbp]
  \centering
% \begin{tikzpicture}[scale=0.5, transform shape]
  \begin{tikzpicture}[scale=0.5, transform shape, every node/.style={font=\fontsize{19pt}{22pt}\selectfont}]

  \draw[white] (-5,0) -- (5,0) node[right, white] {$x$};
  \draw[white] (0,0) -- (0,5) node[above, white] {$y$};

  \def\mean{0}
  \def\sigma{1}
  \def\A{2}
  \draw[domain=-2.6:2.6, samples=100] plot (\x, {\A * exp(-(\x-\mean)^2 / (2*\sigma^2))}) node[right] {$\widehat{\pi(\alpha|Y)}$};

  \begin{scope}[shift={(0, 5)}]
    
    \draw[fill=black!10, rotate=45] (0,0) ellipse ({4/sqrt(2)} and {2/sqrt(2)});
    \draw[fill=black!20, rotate=45] (0,0) ellipse ({3/sqrt(2)} and {1.5/sqrt(2)});
    \draw[fill=black!30, rotate=45] (0,0) ellipse ({2/sqrt(2)} and {1/sqrt(2)});
    \draw[fill=black!40, rotate=45] (0,0) ellipse ({1/sqrt(2)} and {0.5/sqrt(2)});
    
    \draw[dashed] (-3,-2.3) -- (3,2.3) node[right] {``The curve corresponding to $(X^*_{\alpha}, \alpha)$'' };
  \end{scope}

  \node at (2.4,4.5) {$\pi(X, \alpha|Y)$};
  \node at (3.5,3.0) {$(\alpha, \pi(X^*, \alpha|Y))$};
\draw[->] (1.5,3.) -- (0.15,2.05);

  \node at (5.6,1.7) {$(\alpha+\eta, \pi(X^*_{\alpha = \alpha + \eta}, \alpha+\eta|Y))$};
    \draw[->] (1.9,1.75) -- (0.65,1.76);

  \draw[dashed] (0.55,5.2) -- (0.55,{\A * exp(-(\mean+0.5)^2 / (2*\sigma^2))});
  \draw[dashed] (-2.2,3) -- (-2.2,{\A * exp(-(\mean-2)^2 / (2*\sigma^2))});
  \draw[dashed] (0,5) -- (0,{\A * exp(-(\mean)^2 / (2*\sigma^2))});

\end{tikzpicture}
  \caption{Illustration of the Variance Estimation Method in a Bivariate Distribution}
  \label{fig:var_estimation}
\end{figure}
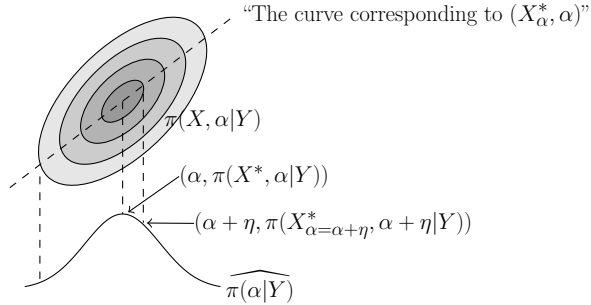

To extend this estimation to all parameters of interest, we repeatedly apply Algorithm \ref{al2} individually to each parameter. Since we are able to recover the covariance matrix, we are also able to estimate variance of functions of parameters. As an illustration, let us look for example into \(\text{Var}(\gamma^w - \gamma^b)\). This estimation is achieved through a sequential application of Algorithm \ref{al2} to each component parameter. Initially, the algorithm is applied to \(\gamma^w\), yielding estimates \(\widehat{\text{Var}(\gamma^w|Y)}\) and \(\widehat{\text{Cov}(\gamma^w, \gamma^b|Y)}\). Following this, a similar application to \(\gamma^b\) provides \(\widehat{\text{Var}(\gamma^b|Y)}\) and an alternative estimate of \(\widehat{\text{Cov}(\gamma^w, \gamma^b|Y)}\). The variance of the composite parameter \(\gamma^w - \gamma^b\) is then approximated using the formula \(\widehat{\text{Var}(\gamma^w|Y)} + \widehat{\text{Var}(\gamma^b|Y)}-2\widehat{\text{Cov}(\gamma^w, \gamma^b|Y)}\). Given that the function is not exactly quadratic, discrepancies between the two covariance estimates may arise. To mitigate this, one can proceed with averaging these estimates for a more robust estimation.
\section{Simulation studies}\label{S:Simulation }
We simulate two networks, one  with $n=100$ actors and one with $n=1000$ actors. For each one of them, we simulate two parameter settings, one for positive attraction or `flocking' ($\gamma^b \geq 0$) and the other for repulsion ($\gamma^b \leq 0$) or `polarization' between two groups over time. For each setting we simulate $20$
data sets with the number of time points $T = 10$.
\begin{table}[H]
    \centering
    \resizebox{\columnwidth}{!}{
\begin{tabular}{ cccccccccc }
&\multicolumn{2}{c}{$\alpha=1$} & \multicolumn{2}{c}{$\delta=2$} & \multicolumn{2}{c}{$\gamma^w=0.25$} & \multicolumn{2}{c}{$\gamma^b=0.5$} \\
\hline 
&$\hat{\alpha}$ & $\widehat{\text{Var}(\alpha)}$ & $\hat{\delta}$ & $\widehat{\text{Var}(\delta)}$ & $\hat{\gamma^w}$ & $\widehat{\text{Var}(\gamma^w)}$ & $\hat{\gamma^b}$ & $\widehat{\text{Var}(\gamma^b)}$ \\
\hline 
n=100&0.817 (0.029) & 0.025 (0.001) & 1.934 (0.023) & 0.024 (0.001) & 0.202 (0.135) & 0.125 (0.026) & 0.489 (0.14) & 0.117 (0.011) \\
\hline
n=1000&0.972 (0.003) & 0.003 ($<$0.001) & 1.992 (0.003) & 0.002 ($<$0.001) & 0.25 (0.028) & 0.03 (0.001) & 0.498 (0.030) & 0.028 (0.001)\\
\hline
&\multicolumn{2}{c}{$\alpha=1$} & \multicolumn{2}{c}{$\delta=3$} & \multicolumn{2}{c}{$\gamma^w=0.45$} & \multicolumn{2}{c}{$\gamma^b=-0.5$} \\
\hline 
&$\hat{\alpha}$ & $\widehat{\text{Var}(\alpha)}$ & $\hat{\delta}$ & $\widehat{\text{Var}(\delta)}$ & $\hat{\gamma^w}$ & $\widehat{\text{Var}(\gamma^w)}$ & $\hat{\gamma^b}$ & $\widehat{\text{Var}(\gamma^b)}$ \\
\hline 
n=100&0.825 (0.043)&0.038(0.001) & 2.868 (0.045) & 0.038 (0.001)  & 0.302 (0.041) & 0.035 (0.003)& -0.54 (0.035) & 0.03 (0.004)\\
\hline
n=1000&0.971 (0.003) & 0.004 ($ <0.001$) & 2.976 (0.003) & 0.004 ( $<0.001$) & 0.433 (0.018) & 0.015 (0.004) & -0.509 (0.017)&0.014(0.004)\\
\hline
\end{tabular}}
    \caption{Posterior-based mean (empirical standard deviation) for point estimation (Algorithm \ref{al1}) and variance estimation (Algorithm \ref{al2}) in flocking and polarization settings, based on $n = 100$ and $n = 1000$ nodes, with $T = 10$ time points.  
    }
    \label{tab:sim}
\end{table}

We want to compare the mean of the point estimates over $20$ simulations with the truth and we compare the mean of the estimated variances with the standard deviation of the point estimates over $20$ simulations.

With the number of nodes $n = 100$, Table \ref{tab:sim} shows that  the point estimates are reasonably accurate
compared with the truth. In the flocking setting, all estimates are slightly biased to be smaller than the truth. In the polarization setting, all estimates except  for repulsion between groups ($\gamma^b$) are slightly biased to be smaller than the truth. The variance estimates are reasonably accurate and are generally slightly smaller than the true standard deviation (See Table \ref{tab:sim}). The empirical standard deviations are obtained by 20 MCMC trials in each parameter setting. (See Table \ref{tab:sim}, where the empirical standard deviations are shown in parenthesis).

Table \ref{tab:sim} shows that  when the number of nodes is $n = 1000$, the point estimates are more accurate than when the number of nodes $n = 100$, both in the flocking setting and the polarization setting, and similarly for the variance estimates.

We notice that the estimation of the standard deviation slightly underestimates the true standard deviation. One reason is that the posterior function can have more than one local minimum. This class of variation estimation methods is able to approximate the contribution of one local minimum to the total variance. Note that multiple local minima contribute to total variance, but in our experiments this did not result into serious issues.

As a comparison of the computational efficiency of the proposed SGD algorithm with MCMC  algorithm, in the X data analysis in the following section, on the reduced data set with 200 nodes, the MCMC took about four hours (CPU) while the proposed SGD-based algorithm only took fewer than 10 minutes to run (GPU). The supplementary file includes a more thorough comparison among MCMC and SGD-based estimation (CPU and GPU).

In the supplementary file, we also include a simulation study to demonstrate the algorithm's robustness in the scenario where nodes are allowed to leave the network over time, i.e., in the case where the vertex set $V_t$ changes dynamically in time.

\section{X Data Analysis}\label{S:TwitterDataAnalysis}
\subsection{Exploratory Analysis}
The X Congressional Hashtag Networks dataset in \cite{10.1093/jrsssa/qnad008} was based on 843,907 tweets from 796 US Congresspersons’ accounts from 2010-2020. Yearly binary networks were built from the tweets. Nodes represent members and edges show whether two members used common hashtags more than the year's average. The nodes in the yearly networks vary as Congress members' X participation changes due to reelections, late joins, or early departures see Figure \ref{count}. Figure \ref{density} shows temporal evolution of edge density within the Democratic and Republican parties individually, alongside the inter-party edge density, as well as the edge density across the entire network. Similar trends were observed in \cite{10.1093/jrsssa/qnad008}: connections increased for the first four years for all types of edges. From 2015 to 2020, edge density within the Democratic party continues to rise, while inter-party edge density decays slowly and edge density within the Republican party drops sharply.

\begin{figure}[h]
\centering
\includegraphics[width=0.5\textwidth]{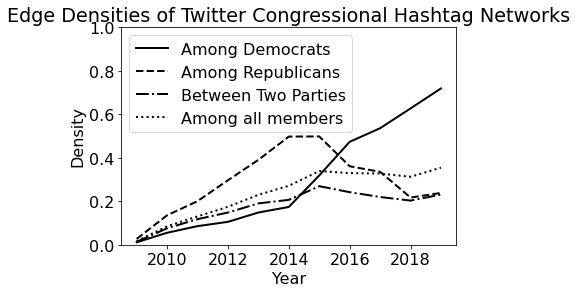}
\caption{Temporal evolution of edge density within the Democratic and
Republican parties, the inter-party edge density and the edge density among all members}
\label{density}
\end{figure}
Figure \ref{count} shows numbers of re-elected and newly elected Democratic congressman and Republican congressman in the X network. We can see an upward trend in the number of X users in both parties. The number of Democratic X users grows more steadily with fewer fluctuation compared with the Republican X users. Notably, Republicans outnumber Democrats in the X network for 9 out of 11 years.
\subsection{Results}

\begin{figure}[h]\label{Fig:PointEstimates}
\centering
\includegraphics[width=0.75\textwidth]{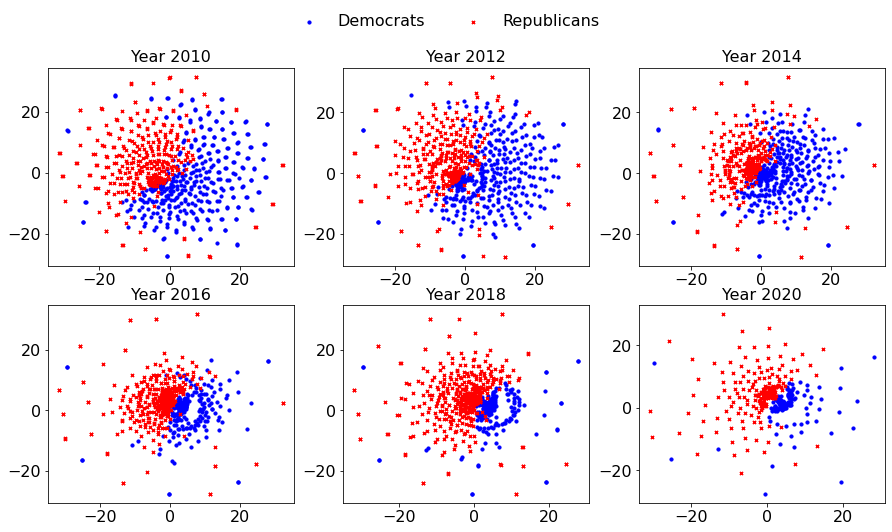}
\caption{ The point estimates of the latent positions with the extended model. The dots that mainly occupy the left portion of each plot represent Democrats and the dots that mainly occupy the right portion of each plot represent Republicans.
 At each time point, the two parties consistently occupy different halves of the space. The Democrats exhibited clustering behavior throughout the period. Conversely, the Republicans initially displayed a similar flocking behavior, but gradually began to disperse around the year 2016.}
\label{positionplot}
\end{figure}

\begin{table}
\centering
\begin{tabular}{*{6}{c}}
 & $\hat \alpha$ & $\hat \delta$ & $\hat \gamma_1^w$ & $\hat \gamma_2^w$ & $\hat \gamma^b$\\\hline
Mean & 3.2 & 1.08 & 0.34 & -0.11 & -0.22 \\
SD & 0.011 & 0.010 & 0.017 & 0.013 & 0.0085 \\
\hline
\end{tabular}
\caption{ 1-Democrats, 2-Republicans.The point estimates and the posterior standard deviation of the global parameters in the extended model fitted to the full X data set spanning from 2010 to 2020.  The between-group attraction is -0.22, indicating that there is polarization between the two parties. The within group coefficient is 0.34 for Democrats, and -0.11 for Republicans, indicating that the Democrats were flocking, while the Republicans were polarizing within their own group.}
\label{tab_para_twitter}
\end{table}
We fit our model with time-invariant parameters to the full X data set from year 2010 to 2020. The AUC values computed at each year are all above 0.974, and the overall AUC value computed across all times is 0.985, providing evidence that our model fits the data very well. The point estimates and the posterior standard deviation of model parameters are provided in Table \ref{tab_para_twitter}. Point estimates were produced using the SGD methodology of Section \ref{SS:PointEstimate} and standard deviation estimates were produced using the methodology of Section \ref{SS:VarianceEstimate}. In the supplementary file, we present two diagnostic tests demonstrating that the underlying assumptions of the variance estimation method proposed in Section \ref{SS:VarianceEstimate} are largely accurate for this dataset, i.e., the diagnostic tests do not show evidence of violating the underlying assumptions of the variance estimation method.

The point estimates of persistence parameter is 1.08, implying that there is a higher likelihood for it to also appear at time $t$ if an edge is present at time $t-1$. The between-group attraction coefficient $\gamma^{b}$ is -0.22, indicating that there is polarization between the two parties. The within group coefficients $\gamma^{w}$ are 0.34 for Democrats, and -0.11 for Republicans, indicating that the Democrats were flocking, while the Republicans were polarizing within their own group. Importantly, we note that in the original model in \cite{10.1093/jrsssa/qnad008}, only members that are present at all times are kept. The parameter estimation for attractions within group indicated that while both parties have moved away from one another, they generally flocked to their own. By incorporating all members into the extended model, we effectively eliminate the potential for selection bias, ensuring a more accurate and comprehensive analysis of the data. As a result of this approach, we have uncovered that the force within the Republican group is negative, in contrast to \cite{10.1093/jrsssa/qnad008}.

Figure \ref{positionplot} shows the point estimates of latent positions for each member of Congress in the X hashtag networks. At each time, two parties consistently occupy different halves of the space. The Democrats were flocking during the whole time. Conversely, the Republicans initially displayed a similar flocking behavior, but gradually began to disperse after year 2017.

\begin{figure}[h]
\centering
\includegraphics[width=0.35\textwidth]{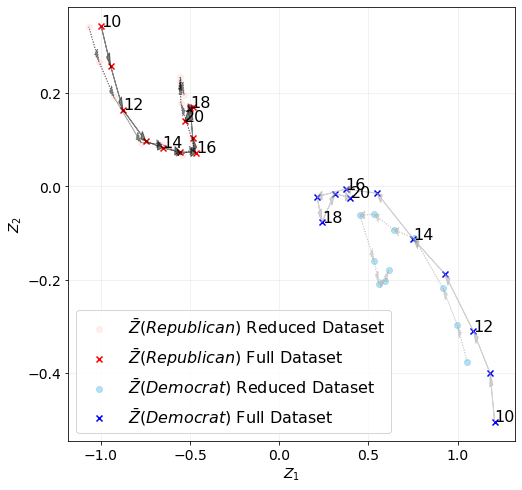}
\caption{The trajectory of the mean of the latent positions of the members in each party. The average position of the two parties gradually and consistently approached each other from 2011 to 2016 and after 2016, both of them took a U-turn and moved apart.}
\label{twomeans}
\end{figure}

Figure \ref{twomeans} shows the trajectory of the mean of the latent positions of the members in each party. The lighter color represents results obtained from the reduced data set used in \cite{10.1093/jrsssa/qnad008}'s original analysis. While the darker color represents results obtained with our model and inference techniques
from the full data set. In both cases, the average position of the two parties gradually and consistently approached each other from 2011 to 2016. But after 2016, both of them took a U-turn and moved apart. However, we see that the analysis based on the full data set suggests that the Democrats were in fact noticeably closer to the Republicans just before 2016 than suggested by the analysis based on the reduced data.

\begin{figure}[h]
\centering
\includegraphics[width=0.5\textwidth]{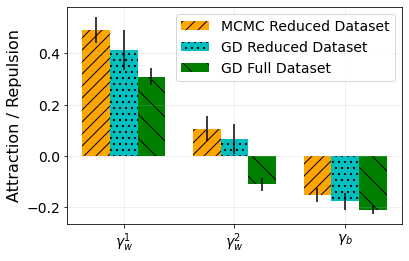}
\caption{Comparison of attraction parameters from MCMC method on the reduced data set and the SGD-based method on the reduced and full data set. The MCMC method and the SGD-based method yield very similar point estimates and variance estimates for the attraction parameters on the reduced data set. On the full data set, the model concludes that there is a large attractive force within the Democratic Party. By contrast, we find that within the Republican Party the force is repulsive.}
\label{compare_attraction}
\end{figure}
Figure \ref{compare_attraction} compares attraction parameters with that found in \cite{10.1093/jrsssa/qnad008}. The MCMC and the SGD-based method yield close point estimates and variance estimates for the attraction parameters on the reduced data set. The point estimates obtained from the MCMC method fall within the 95\% confidence interval of the SGD-based method. On the full data set, the model concludes that there is a large attractive force within the Democratic Party. By contrast, we find that within the Republican Party the force is repulsive.
\begin{figure}[h]
\centering
\includegraphics[width=0.5\textwidth]{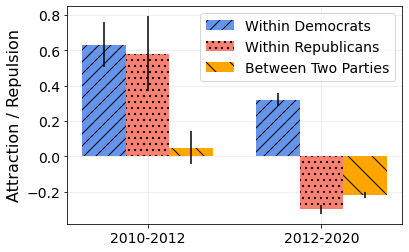}
\caption{Evolution of posterior means and 95\% CI for
Democrats and Republicans, and between-group attraction/repulsion in X Congressional Hashtag Networks, the values of which at each time period in the horizontal axis are
obtained by fitting the model using separately parameterized networks in the corresponding
time periods.}
\label{varying}
\end{figure}
To analyze the changes in attraction, repulsion,
we fit a series of models that allow a change-point to vary from 2011 to 2019. The resulting fitted models with different change-points were compared based on Bayesian Information Criterion (BIC). The one with the lowest BIC value was selected. We identified 2012 as the year of change-point (See supplementary file for more details on this). Figure \ref{varying} illustrates the evolution of within-group attraction and repulsion for Democrats and Republicans, and between-group attraction and repulsion.  The within-Republican coefficient is positive in the first time period and is negative in the second time periods from 2012 to 2020.  This suggests polarization within Republican members started to rise within the period 2012-2020. The within-Democrat coefficient is positive in both time periods, although its magnitude decreases a bit in the second time period from 2012 to 2020. The between-group coefficient (orange bars) is not statistically significant at the 95\% confidence level from 2010 to 2012 and is negative in the second time periods from 2012 to 2020.

The observed patterns of increasing polarization between the Democratic and Republican parties align with findings from previous studies. For instance, \cite{andris2015rise} documented similar trends through the analysis of co-sponsorship patterns in Congress, while \cite{garimella2017long} identified heightened polarization in the retweet networks of U.S. politicians on Twitter. These studies, along with our findings, highlight political polarization in different environments in U.S..

\section{Discussion}
The extended CLSNA model proposed in this paper together with the SGD-based inference method, successfully addressed the challenge of scalability to large networks with nodes entering and leaving the network. In particular, (a) the introduction of the SGD parameter estimation method, (b) the development of the novel variance estimation approach for approximate Bayesian inference which transforms variance estimation into an optimization problem solvable with SGD, and, (c) the extension of the model to allow dynamic node participation, have significantly improved the CLSNA model's utility for analyzing dynamic networks. When applied to longitudinal social networks on X, the model mitigates selection bias and reveals previously concealed negative force within the Republican party.

The SGD parameter estimation method and the variance estimation method are general approaches to approximate Bayesian inference. Therefore, they can be combined with other strategies to make inference faster. For example, they can be used with case-control log-likelihood approximation as in \cite{raftery2012fast}. This can be achieved by using stratified sampling to draw the edge probability terms during the stochastic sampling process.

With the increased efficiency of the proposed method, extending the model to accommodate more than two groups with attractors becomes more feasible. This extension is particularly useful in contexts such as multi-party political systems in Germany and Canada. For a small number of groups, the computational impact is expected to be minimal. Moreover, our library has been developed with this flexibility in mind and allows for such an extension.

In future work we plan to incorporate observed explanatory variable terms into the edge probability. For instance, a linear predictor based on explanatory variables can be integrated into the edge probability function, see \cite{hoff2007modeling}. In this way, we can more effectively separate the latent effect from the fixed effect and accurately capture the repulsive and attractive force among the different actors.

\section*{Acknowledgement}
This research was supported by US NSF-DMS 2311500, NSF SES-2120115, Canadian NSERC RGPIN-2023-03566 and NSERC DGDND-2023-03566.

\section*{Supplementary Materials}
\begin{description}
    \item[Appendix:] Contains proofs of the theoretical results, additional statistical tests and diagnostics to support the numerical results, implementation details of the algorithm, performance/efficiency studies of the algorithm, and a simulation study demonstrating the algorithm's robustness against node turnover. (supp.pdf, PDF)
    \item[Source Code:] The source code is available in the GitHub repository. The current library implementation is specifically designed for a 2-party system. 
    
(https://github.com/KolaczykResearch/SGD4CLSNA, GitHub repository)
\end{description}

\bibliographystyle{rss}
\bibliography{references}

\begin{thebibliography}{16}
\expandafter\ifx\csname natexlab\endcsname\relax\def\natexlab#1{#1}\fi
\expandafter\ifx\csname url\endcsname\relax
  \def\url#1{\texttt{#1}}\fi
\expandafter\ifx\csname urlprefix\endcsname\relax\def\urlprefix{URL: }\fi

\bibitem[{Andris et~al.(2015)Andris, Lee, Hamilton, Martino, Gunning and Selden}]{andris2015rise}
Andris, C., Lee, D., Hamilton, M.~J., Martino, M., Gunning, C.~E. and Selden, J.~A. (2015) The rise of partisanship and super-cooperators in the us house of representatives.
\newblock \textit{PloS one}, \textbf{10}, e0123507.

\bibitem[{Garimella and Weber(2017)}]{garimella2017long}
Garimella, V. R.~K. and Weber, I. (2017) A long-term analysis of polarization on twitter.
\newblock In \textit{Proceedings of the International AAAI Conference on Web and social media}, vol.~11, 528--531.

\bibitem[{Hinton(2012)}]{hinton2012neural}
Hinton, G. (2012) Neural networks for machine learning.
\newblock Coursera.
\newblock \urlprefix\url{https://www.coursera.org/course/neuralnets}.

\bibitem[{Hoff(2007)}]{hoff2007modeling}
Hoff, P. (2007) Modeling homophily and stochastic equivalence in symmetric relational data.
\newblock \textit{Advances in Neural Information Processing Systems}, \textbf{20}.

\bibitem[{Hoff et~al.(2002)Hoff, Raftery and Handcock}]{hoff2002latent}
Hoff, P.~D., Raftery, A.~E. and Handcock, M.~S. (2002) Latent space approaches to social network analysis.
\newblock \textit{Journal of the American Statistical Association}, \textbf{97}, 1090--1098.

\bibitem[{Liu and Chen(2021)}]{liu2021variational}
Liu, Y. and Chen, Y. (2021) Variational inference for latent space models for dynamic networks.
\newblock \textit{arXiv preprint arXiv:2105.14093}.

\bibitem[{Matias and Miele(2017)}]{matias2017statistical}
Matias, C. and Miele, V. (2017) Statistical clustering of temporal networks through a dynamic stochastic block model.
\newblock \textit{Journal of the Royal Statistical Society. Series B (Statistical Methodology)}, \textbf{79}, 1119--1141.

\bibitem[{McCarty et~al.(2016)McCarty, Poole and Rosenthal}]{mccarty2016polarized}
McCarty, N., Poole, K.~T. and Rosenthal, H. (2016) \textit{Polarized America: The dance of ideology and unequal riches}.
\newblock mit Press.

\bibitem[{Polyak(1964)}]{polyak1964some}
Polyak, B.~T. (1964) Some methods of speeding up the convergence of iteration methods.
\newblock \textit{USSR Computational Mathematics and Mathematical Physics}, \textbf{4}, 1--17.

\bibitem[{Raftery et~al.(2012)Raftery, Niu, Hoff and Yeung}]{raftery2012fast}
Raftery, A.~E., Niu, X., Hoff, P.~D. and Yeung, K.~Y. (2012) Fast inference for the latent space network model using a case-control approximate likelihood.
\newblock \textit{Journal of Computational and Graphical Statistics}, \textbf{21}, 901--919.

\bibitem[{Rue et~al.(2009)Rue, Martino and Chopin}]{rue2009approximate}
Rue, H., Martino, S. and Chopin, N. (2009) Approximate bayesian inference for latent gaussian models by using integrated nested laplace approximations.
\newblock \textit{Journal of the Royal Statistical Society: Series B (statistical methodology)}, \textbf{71}, 319--392.

\bibitem[{Sarkar and Moore(2005)}]{sarkar2005dynamic}
Sarkar, P. and Moore, A.~W. (2005) Dynamic social network analysis using latent space models.
\newblock \textit{SIGKDD Explorations}, \textbf{7}, 31--40.

\bibitem[{Sewell and Chen(2015{\natexlab{a}})}]{sewell2015analysis}
Sewell, D.~K. and Chen, Y. (2015{\natexlab{a}}) Analysis of the formation of the structure of social networks by using latent space models for ranked dynamic networks.
\newblock \textit{Journal of the Royal Statistical Society: Series C: Applied Statistics}, 611--633.

\bibitem[{Sewell and Chen(2015{\natexlab{b}})}]{sewell2015latent}
--- (2015{\natexlab{b}}) Latent space models for dynamic networks.
\newblock \textit{Journal of the American Statistical Association}, \textbf{110}, 1646--1657.

\bibitem[{Sewell and Chen(2016)}]{sewell2016latent}
--- (2016) Latent space models for dynamic networks with weighted edges.
\newblock \textit{Social Networks}, \textbf{44}, 105--116.

\bibitem[{Zhu et~al.(2023)Zhu, Caliskan, Christenson, Spiliopoulos, Walker and Kolaczyk}]{10.1093/jrsssa/qnad008}
Zhu, X., Caliskan, C., Christenson, D.~P., Spiliopoulos, K., Walker, D. and Kolaczyk, E.~D. (2023) {Disentangling positive and negative partisanship in social media interactions using a coevolving latent space network with attractors model}.
\newblock \textit{Journal of the Royal Statistical Society Series A: Statistics in Society}.
\newblock \urlprefix\url{https://doi.org/10.1093/jrsssa/qnad008}.
\newblock Qnad008.

\end{thebibliography}

%%%%%Supplementary
\newpage

\begin{center}
\def\spacingset#1{\renewcommand{\baselinestretch}%
{#1}\small\normalsize} \spacingset{1}

\if0\blind
{
  \title{\bf  Supplementary  Information for `Stochastic gradient descent-based inference for dynamic network models with attractors'}
  \author{\\
    Hancong Pan\\
    Department of Mathematics and Statistics, Boston University\\
    and \\
    Xiaojing Zhu\\
    Department of Mathematics and Statistics, Boston University\\
    and \\
    Cantay Caliskan \\
    Goergen Institute for Data Science, University of Rochester\\
    and \\
    Dino P. Christenson \\
    Department of Political Science, Washington University in St. Louis\\
    and \\
    Konstantinos Spiliopoulos \\
    Department of Mathematics and Statistics, Boston University\\
    and \\
    Dylan Walker \\
    Argyros School of Business and Economics, Chapman University\\
    and \\
    Eric D. Kolaczyk \\
    Department of Mathematics and Statistics, McGill University\\
    E-mail: eric.kolaczyk@mcgill.ca
  }
  \maketitle
} \fi

 \end{center}

%\date{}
%\doublespacing

%\begin{document}

%\maketitle

%\newpage
\spacingset{1.75} % DON'T change the spacing!

%\date{}
\doublespacing

This supplementary document is organized as follows. In Section \ref{S:Gaussian} we present the proofs of the theoretical results that appear in the main body of the paper. In Section \ref{S:StatisticalAnalysis} we present some additional statistical tests and diagnostics to back up the validity of the numerical results reported in the main body of the paper. In Section \ref{S:ImplementationDetails}, we provide implementation details for the algorithm. In Section \ref{S:AlgorithmComparison}, we compare computational times among MCMC and the newly proposed SGD-based algorithm (in both CPU and GPU settings). We also comment upon the impact of hyper-parameters choices on convergence times. In Section \ref{S:VaryingPercentageLeaving} we present a simulation study to illustrate the algorithm's robustness when nodes exit the network at different turnaround levels.

\section{Some Useful Results of Multivariate Normal Distribution}
\label{S:Gaussian}

Our variance estimation method exploits the properties of the conditional distribution of a multivariate Gaussian distribution. Let \( x \) follow a multivariate normal distribution \( x \sim \mathcal{N}(\mu, \Sigma) \), and hence the conditional distribution of a subset vector \( x_1 \), given its complement vector \( x_2 \), is also a multivariate normal distribution \( x_1 | x_2 \sim \mathcal{N}(\mu_{1|2}, \Sigma_{1|2}) \) with the conditional mean and covariance given by \(\mu_{1|2} = \mu_1 + \Sigma_{12} \Sigma_{22}^{-1} (x_2 - \mu_2)\) and \(\Sigma_{1|2} = \Sigma_{11} - \Sigma_{12} \Sigma_{22}^{-1} \Sigma_{21}\), respectively. The block-wise mean and covariance matrices are defined as \(\mu = \left[\begin{smallmatrix} \mu_1 \\ \mu_2 \end{smallmatrix}\right]\) and \(\Sigma = \left[\begin{smallmatrix} \Sigma_{11} & \Sigma_{12} \\ \Sigma_{21} & \Sigma_{22} \end{smallmatrix}\right]\).

More specifically, the covariance of the conditional distribution \(\text{Cov}(x_1|x_2) \) is a constant as a function of \( x_2 \), indicating that the shape of the conditional distribution is independent of the specific value of \( x_2 \). Additionally, writing \( p(x_1|x_2) \) as the p.d.f of the conditional distribution of \( x_1|x_2 \), the mean \( \mu_{1|2} = \mu_1 + \Sigma_{12} \Sigma_{22}^{-1} (x_2 - \mu_2) \) is the median, and mode of the distribution. The maximum of the p.d.f. function \( p(x_1=\mu_{1|2}|x_2) \) is a constant as a function of \( x_2 \).

Therefore, we can recover the shape of the marginal distribution of $x_2$ from the joint distribution by varying $x_2$ and following the line $x_1=\mu_{1|2}$.

\begin{theorem}
\label{main_theo}
    Given the multivariate Gaussian distribution setting presented above, write $p(x_1, x_2)$ as the p.d.f of the joint distribution of $x_1, x_2$. $p(x_2)$ as the p.d.f of the marginal distribution of $x_2$. Then, we have $p(x_2) \propto p(x_1=\mu_{1|2},x_2)$.
\end{theorem}

In particular,  the marginal distribution is proportional to the joint distribution evaluated on the curve of the conditional mean of $x_1$ given $x_2$.
\begin{proof}
Let us explicitly state the conditional mean and covariance matrix for our given multivariate Gaussian distribution. The conditional mean of \(x_1\) given \(x_2\), denoted \( \mu_{1|2} \), is:
\[ \mu_{1|2} = \mu_1 + \Sigma_{12} \Sigma_{22}^{-1} (x_2 - \mu_2). \]

Additionally, the conditional covariance matrix, denoted \( \Sigma_{1|2} \), is described as:
\[ \Sigma_{1|2} = \Sigma_{11} - \Sigma_{12} \Sigma_{22}^{-1} \Sigma_{21}. \]

The joint distribution, \( p(x_1, x_2) \), for the multivariate Gaussian is:
\[ p(x_1, x_2) = \frac{1}{(2\pi)^{\frac{k}{2}}|\Sigma|^{\frac{1}{2}}} \exp\left(-\frac{1}{2}(x-\mu)^T \Sigma^{-1} (x-\mu)\right). \]

The conditional distributions \( p(x_1 | x_2) \) is:
\[ p(x_1 | x_2) = \frac{1}{(2\pi)^{\frac{n}{2}}|\Sigma_{1|2}|^{\frac{1}{2}}} \exp\left(-\frac{1}{2}(x_1-\mu_{1|2})^T \Sigma_{1|2}^{-1} (x_1-\mu_{1|2})\right). \]

By the definition of conditional probabilities, we have:
\[ p(x_2) = \frac{p(x_1, x_2)}{p(x_1 | x_2)}. \]

Inserting \( x_1 = \mu_{1|2} \) into this relation, we get:
\[ p(x_2) = \frac{p(\mu_{1|2}, x_2)}{p(\mu_{1|2} | x_2)}. \]

Since $p(\mu_{1|2} | x_2)=\frac{1}{(2\pi)^{\frac{n}{2}}|\Sigma_{1|2}|^{\frac{1}{2}}}$ and  \( \Sigma_{1|2} \) is constant as a function of \( x_2 \), we get that
\[ p(x_2) = C p(\mu_{1|2}, x_2). \]
where \(C = (2\pi)^{\frac{n}{2}}| \Sigma_{11} - \Sigma_{12} \Sigma_{22}^{-1} \Sigma_{21} |^{\frac{1}{2}}\), a constant as a function of \(x_2\).
\end{proof}

Theorem \ref{main_theo} provides an expression for the shape of the marginal distribution from the joint distribution, when the joint distribution is a multivariate normal distribution.

Next, we proceed with an important result demonstrating that for a normal distribution the ideas of Theorem \ref{main_theo} yield a formula for the variance of $x_2$ that we can then be turned into a practical SGD-based algorithm for uncertainty quantification. In particular, recall that in a univariate normal distribution, two key parameters - the mean and variance - describe the distribution fully. By considering two distinct points from the distribution, we can deduce these parameters. 

\begin{corollary}
\label{maincol}
   Consider the multivariate Gaussian distribution setting presented above, and let $x_2$ be scalar. Then, for some fixed $\tilde{x}_2\neq \mu_2$, 
\begin{equation}
    \text{Var}(x_2)= \frac{1}{2}\frac{(\mu_2-\tilde{x}_2)^2}{\log p(x_1=\mu_1,x_2=\mu_2)-\log p(x_1=\mu_{1|2},x_2=\tilde{x}_2)}
    \label{Eq:VarianceFormula1}
\end{equation}
\end{corollary}
\begin{proof}
Recall that the p.d.f of a univariate Gaussian distribution \( x_2 \sim \mathcal{N}(\mu_2, \text{Var}(x_2)) \) is given by:
\begin{equation}
p(x_2) = \frac{1}{\sqrt{2\pi \text{Var}(x_2)}} \exp\left(-\frac{(x_2-\mu_2)^2}{2 \text{Var}(x_2)}\right).
\end{equation}
Taking the logarithm, we obtain:
\begin{equation}
\log p(x_2) = -\frac{1}{2} \log(2\pi \text{Var}(x_2)) - \frac{(x_2-\mu_2)^2}{2 \text{Var}(x_2)}.
\end{equation}

From Theorem \ref{main_theo}, we deduce that \( \log p(x_2) \) and \( \log p(x_1=\mu_{1|2},x_2) \) differ by a constant since \( p(x_2) \propto p(x_1=\mu_{1|2},x_2) \). Considering two distinct points, \( x_2 = \mu_2 \) and \( x_2 = \tilde{x}_2 \):

For \( x_2 = \mu_2 \):
\begin{equation}
\log p(x_2=\mu_2) = -\frac{1}{2} \log(2\pi \text{Var}(x_2)).
\end{equation}

For \( x_2 = \tilde{x}_2 \):
\begin{equation}
\log p(x_2=\tilde{x}_2) = -\frac{1}{2} \log(2\pi \text{Var}(x_2)) - \frac{(\tilde{x}_2-\mu_2)^2}{2 \text{Var}(x_2)}.
\end{equation}

Define the difference in these logarithmic probabilities as:
\begin{equation}
\Delta = \log p(x_1=\mu_1,x_2=\mu_2) - \log p(x_1=\mu_{1|2},x_2=\tilde{x}_2).
\end{equation}

Using the above, we get:
\begin{equation}
\Delta = \frac{(\mu_2-\tilde{x}_2)^2}{2 \text{Var}(x_2)}.
\end{equation}

Rearranging gives:
\begin{equation}
\text{Var}(x_2) = \frac{(\mu_2-\tilde{x}_2)^2}{2(\log p(x_1=\mu_1,x_2=\mu_2) - \log p(x_1=\mu_{1|2},x_2=\tilde{x}_2))},
\end{equation}
completing the proof of the corollary.
\end{proof}
Beyond just the variance represented by the diagonal elements of the covariance matrix, there's also a formula for the off-diagonal elements. This allows us to reconstruct a single column of the covariance matrix.
\begin{corollary}
\label{secondcol}
 Consider the multivariate Gaussian distribution setting presented above. Then, for some fixed $\tilde{x}_2\neq \mu_2$,
    \[ \Sigma_{12} =  \frac{\Sigma_{22} }{ (\tilde{x}_2 - \mu_2)}(\mu_{1|2} -\mu_1), \]
    where \( \mu_{1|2} \) denotes the conditional mean of \(x_1\) given \(x_2 = \tilde{x}_2\).
\end{corollary}

The methodology established for a generic normal distribution \( p(x_1, x_2) \) is adapted to estimate variance in the posterior distribution \( \pi(\cdot|Y) \), effectively applying the same theoretical concepts to the Bayesian analysis under a Laplace approximation.

\section{Supplementary Results in X Data Analysis}\label{S:StatisticalAnalysis}
\subsection{BIC values}
In our analysis, we utilized the Bayesian Information Criterion (BIC) to determine the best change-point among the competing models. For the X platform dataset, we evaluated the BIC values for potential single change-points from 2012 onwards, as presented in Table below. The model suggesting a change-point in 2012 yielded the lowest BIC value.
\begin{table}[H]
    \centering
    \begin{tabular}{*{9}{c}}

         &2012 &2013&2014&2015&2016&2017&2018&2019 \\
         \hline
         
       BIC  &\textbf{220138}&220163&220201&220189&220212&220203&220234&220224\\
       \hline
    \end{tabular}
    \caption{BIC values for competing models with different change-point for the X platform data.}
    \label{tab:my_label}
\end{table}

\subsection{Diagnostics}
Note that the variance estimation algorithm is similar to a quadratic approximation to the log-posterior density. We performed diagnostics to evaluate the assumption that the log likelihood function can be approximated by a quadratic function in a small neighborhood centered around the mode.
\subsubsection{Normality test}
We rewrite equation \eqref{Eq:VarianceFormula1} as follows:
\begin{align}
       \max_{\theta_1 = \theta_1^{*}-\eta} \log \pi(Z,\theta|Y) =  \max_{\theta_1 = \theta_1^{*}} \log \pi(Z,\theta|Y) -  \frac{1}{2}\frac{\eta^2}{
         \widehat{\text{Var}(\theta_1|Y)}},
\end{align}
and if our assumptions are correct, then $\widehat{\text{Var}(\theta_1|Y)}$ should be a constant no matter the choice of $\eta$.  We vary $\eta$ to obtain a plot of $\max_{\theta_1 = \theta_1^{*}-\eta} \log \pi(Z,\theta|Y)$ as a function of $\eta$ to test whether this assumption is valid. If this assumption holds, $\max_{\theta_1 = \theta_1^{*}-\eta} \log \pi(Z,\theta|Y)$ should be approximately a quadratic function as a function of $\eta$.

We choose $\theta_1$ to be $\gamma^{w}_R$, and set $\eta$ to be {-0.03, -0.02, -0.01, 0.01, 0.02, 0.03}. In Figure \ref{test1} we plot $\max_{\theta_1 = \theta_1^{*}-\eta} \log \pi(Z,\theta|Y)$ as a function of $\eta$ (left) and $\eta^2$ (right). We can see that $\max_{\theta_1 = \theta_1^{*}-\eta} \log \pi(Z,\theta|Y)$ is approximately quadratic as a function of $\eta$ and linear as a function of $\eta^2$. The diagnostic plot does not show evidence of violating the underlying assumptions of the variance estimation method.
\begin{figure}[h]
\centering
\includegraphics[width=\textwidth]{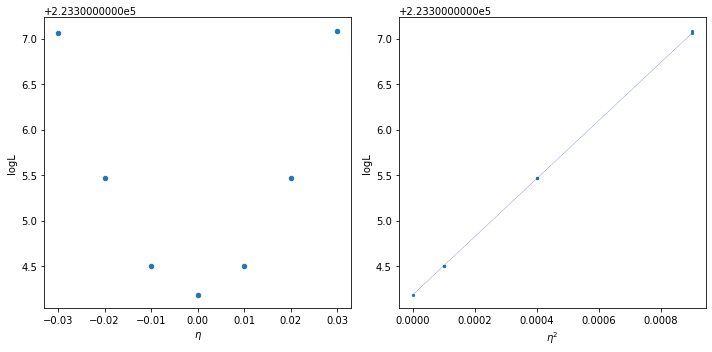}
\caption{We choose $\theta_1$ to be $\gamma^{w}_R$, and set $\eta$ to be {-0.03, -0.02, -0.01, 0.01, 0.02, 0.03}. We plot $\max_{\theta_1 = \theta_1^{*}-\eta} \log \pi(Z,\theta|Y)$ as a function of $\eta$ (left) and $\eta^2$ (right).
}
\label{test1}
\end{figure}

\subsubsection{Linearity test}
\begin{figure}[h]
\centering
\includegraphics[width=\textwidth]{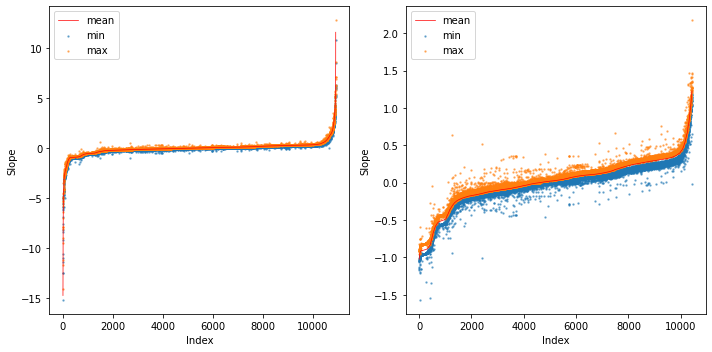}
\caption{We choose $\theta_1$ to be $\gamma^{w}_R$, and set $\eta$ to be {-0.03, -0.02, -0.01, 0.01, 0.02, 0.03}. We calculate $\frac{\mu(X_j|X_1=\mu_{1}+\eta)-\mu_j}{\eta}$ for each latent position parameters under each $\eta$. The latent position parameters are ordered by the mean of its slopes and plot the mean maximum and minimum of its slopes (left) and the the slopes of 95\% of the latent position parameters after removing the tail 5\% that have a very large or very small slopes. 
}
\label{test2}
\end{figure}
From Theorem \ref{main_theo}, we see that the conditional mean of  $X_{2:n}|X_1$ is linear as a function of $X_1$ under the Gaussian assumption, i.e.,
\begin{align}
    \mu(X_{2:n}|X_1=\mu_{1}+\eta) = \mu_{2:n}+\Sigma_{21}\Sigma_{11}^{-1}\eta.
\end{align}

We vary $\eta$ to obtain the slope $\frac{\mu(X_j|X_1=\mu_{1}+\eta)-\mu_j}{\eta}$ for each $X_j$ with different $\eta$ to test whether this assumption is valid. If this assumption holds, $\frac{\mu(X_j|X_1=\mu_{1}+\eta)-\mu_j}{\eta}$ should be a constant regardless the choice of $\eta$ for each $X_j$. 

We use the results of the previous test, set $\theta_1$ to be $\gamma^{w}_R$, and $\eta$ to take values in the set $\{-0.03, -0.02, -0.01, 0.01, 0.02, 0.03\}$. We calculate $\frac{\mu(X_j|X_1=\mu_{1}+\eta)-\mu_j}{\eta}$ for each latent position parameters and each $\eta$. In Figure \ref{test2} we order the latent position parameters by the mean of its slopes and plot the mean maximum and minimum of its slopes (left) and the the slopes of 95\% of the latent position parameters after removing the tail 5\% that have a very large or very small slopes.  We can see that the slopes, or the correlated changes for each latent position parameters when we change $\gamma^w_R$ are close under different choice of $\eta$. Therefore the diagnostic plot does not show substantial evidence of violating the underlying assumptions of the variance estimation method.
\section{Implementation Details}\label{S:ImplementationDetails}
We use gradient descent with momentum and choose different learning hyperparameters for the latent position parameters and the global parameters. With gradient descent, we use all the terms in the log posterior function instead of taking a sampled posterior function. We stop running gradient descent when the parameter updates drop below a threshold of $\epsilon = 10^{-4}$, a condition we check every 100 gradient steps. Priors for $\alpha$ and $\delta$ were chosen to be $\mathcal{N}(0, 100)$ to keep it flat and uninformative. We chose the priors for $\gamma^{w}$ and $\gamma^b$ to be $\mathcal{N}(0.5, 100)$ and $\mathcal{N}(-0.5, 100)$ to reflect the prior belief of polarization, however these are also quite uninformative given the large variance. We fix $\tau^2$ at 10, $\sigma^2$ at 1 and $\phi^2$ at 10.
We choose \( p = 2 \), following \cite{10.1093/jrsssa/qnad008}. The choice of two dimensions aligns with DW-NOMINATE, a widely recognized model of congressional ideology. This model demonstrates that two dimensions can account for up to 90\% of the variation in roll call voting \citep{mccarty2016polarized}.

For the step size $\lambda$, we started by conducting a search to identify the optimal initial step size for our model by observing the loss curves for various candidates. The goal was to identify a step size that avoided slow convergence and prevented the loss from exploding. Once the best initial step size was determined, we proceeded with training the model. During training, if the loss plateaued, we halved the step size to promote further reduction in the loss. After one such adjustment, further halving did not yield significant improvements, so we concluded the training process at that point.
To choose \(\eta_{\alpha}\), we take \(\alpha\) as an example. We start by selecting an initial guess for \(\eta_{\alpha}\) such that \(\log \pi(\alpha^{*}) - \log \pi(\alpha^{*} + \eta_{\alpha})\) is neither too large nor too small, aiming for a value between 10 and 50. Here, \(\log \pi(\alpha^{*})\) denotes the log probability density function \(\pi\) evaluated at the mode with $\alpha$ set to be \(\alpha^{*}\), while \(\log \pi(\alpha^{*} + \eta_{\alpha})\) represents the same function evaluated at the mode with $\alpha$ set to be \(\alpha^{*}+\eta_{\alpha}\). This range ensures that the perturbation is sizable enough to influence the parameters meaningfully but not so large that they deviate excessively from the optimum. In our case, we opted for a uniform \(\eta_{\alpha}\) for all parameters for which we wanted to estimate variance, and this approach worked well.
\section{Algorithm Comparison and Analysis}\label{S:AlgorithmComparison}
\begin{figure}[h]
    \centering
    \includegraphics[width=0.8\textwidth]{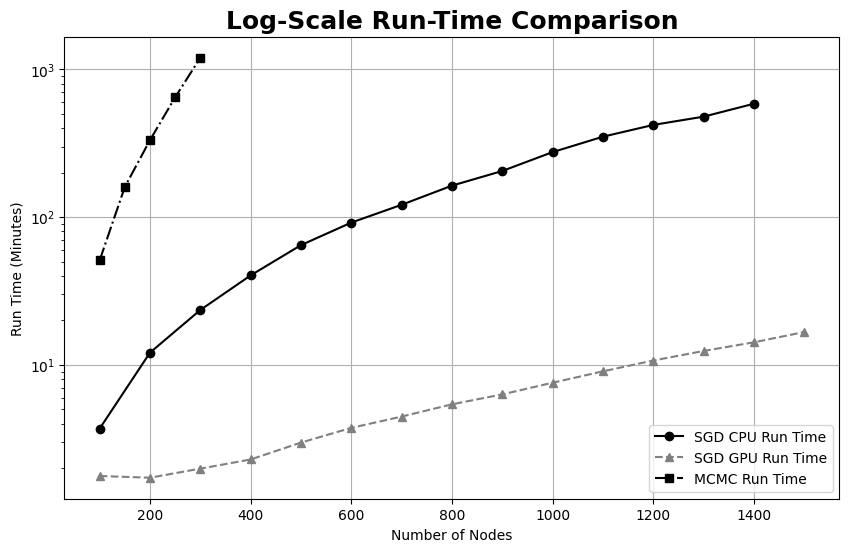}
    \caption{Comparison of Convergence Times for SGD (CPU and GPU) and MCMC Algorithms Across Different Network Sizes.}
    \label{fig:runtime_comparison}
\end{figure}

Figure \ref{fig:runtime_comparison} illustrates the convergence times for SGD-based algorithms executed on a CPU and an Nvidia GPU A40, alongside the convergence times for the MCMC algorithm. The x-axis denotes the number of nodes in the network, while the y-axis represents the convergence time in minutes.
The blue curve corresponds to the convergence times of the SGD algorithm when executed on a CPU, the red curve represents the convergence times of the same algorithm when executed on a GPU, and the green curve shows the convergence times for the MCMC algorithm.
The SGD algorithm provides approximately a 30-fold reduction in convergence time compared to the MCMC algorithm. Additionally, parallelization further improves the scalability of the SGD method: executing the SGD algorithm on the Nvidia GPU A40 yields substantial efficiency gains, resulting in around a 20-fold reduction in convergence time compared to its execution on the CPU.
Note that the reported convergence times for the SGD algorithm include both the point estimation step and the variance estimation step.
\subsection{Impact of Hyperparameters on Convergence Time}
It is important to understand how hyperparameters, such as step size (\(\lambda\)), affect the convergence speed. Choosing the correct step size (\(\lambda\)) and momentum in SGD is crucial, as it significantly affects the convergence behavior. Specifically, if the step size is too large, the algorithm may diverge, causing the loss to explode. If the step size is too small, the algorithm will converge very slowly. To address this, we fixed the momentum parameter at 0.99 and focused on tuning the step size.
For the step size \(\lambda\), we started by conducting an empirical search to identify an optimal initial value. We observed the loss curves for various step size candidates, aiming to find a step size that avoids both slow convergence and loss explosion. After determining the best initial step size, we began training the model. During training, if we noticed that the loss plateaued, we halved the step size. This halving process was repeated until the log posterior function convergences.
Once the correct magnitude is established, fine-tuning the step size only results in marginal gains in convergence speed. For the computational time plot presented above, we used the step size of $0.001$ for all latent position parameters and $0.001$ for all global parameters, while using the sign of the gradient (+1/-1) for updating the global parameters instead of the actual gradient values.
\subsection{Computational Complexity and Factors Contributing to Speedup}
Each iteration of the proposed SGD-based algorithm has a computational complexity of \(O(T \times N^2)\) for both time and space, where \(T\) is the number of time points in the network time-series and \(N\) represents the number of nodes in each network. This complexity comes from the requirement to process all node pairs across each time point.
Each iteration of the MCMC algorithm also has the same computational complexity of \(O(T \times N^2)\). The significant speedup of the SGD-based algorithm predominantly stems from two key factors:
\begin{enumerate}
    \item \textbf{Parallelization on GPU}: The ability to leverage the parallel processing capabilities of GPUs significantly reduces the time spent per iteration. This factor provides around a 20-fold increase in performance compared to the CPU execution.
    \item \textbf{Fewer Iterations to Converge}: The SGD-based algorithm requires far fewer iterations to achieve convergence compared to the MCMC algorithm. This results in approximately a 30-fold increase in performance.
\end{enumerate}

\section{Varying percentage leaving}\label{S:VaryingPercentageLeaving}

\begin{table}[H]
    \centering
    \resizebox{\columnwidth}{!}{
\begin{tabular}{ cccccccccc }
&\multicolumn{2}{c}{$\alpha=1$} & \multicolumn{2}{c}{$\delta=2$} & \multicolumn{2}{c}{$\gamma^w=0.25$} & \multicolumn{2}{c}{$\gamma^b=0.5$} \\
\hline 
&$\hat{\alpha}$ & $\widehat{\text{Var}(\alpha)}$ & $\hat{\delta}$ & $\widehat{\text{Var}(\delta)}$ & $\hat{\gamma^w}$ & $\widehat{\text{Var}(\gamma^w)}$ & $\hat{\gamma^b}$ & $\widehat{\text{Var}(\gamma^b)}$ \\
\hline
Turnover=0\%&0.94 (0.004) & 0.005 ($<$0.001) & 1.98 (0.005) & 0.005 ($<$0.001) & 0.24 (0.025) & 0.028 (0.007) & 0.50 (0.025) & 0.026 (0.008) \\
\hline
Turnover=10\%&0.94 (0.006) & 0.005 ($<$0.001) & 1.98 (0.007) & 0.005 ($<$0.001) & 0.24 (0.032) & 0.030 (0.007) & 0.50 (0.030) & 0.027 (0.008) \\
\hline
Turnover=20\%&0.94 (0.006) & 0.005 ($<$0.001) & 1.98 (0.007) & 0.006 ($<$0.001) & 0.25 (0.041) & 0.032 (0.007) & 0.49 (0.034) & 0.028 (0.008) \\
\hline
Turnover=30\%&0.94 (0.008) & 0.005 ($<$0.001) & 1.98 (0.005) & 0.007 ($<$0.001) & 0.24 (0.044) & 0.034 (0.008) & 0.50 (0.034) & 0.030 (0.010) \\
\hline
Turnover=40\%&0.94 (0.007) & 0.005 ($<$0.001) & 1.98 (0.008) & 0.008 ($<$0.001) & 0.26 (0.042) & 0.036 (0.008) & 0.48 (0.029) & 0.030 (0.009) \\
\hline
Turnover=50\%&0.94 (0.008) & 0.005 ($<$0.001) & 1.98 (0.008) & 0.009 ($<$0.001) & 0.25 (0.042) & 0.038 (0.008) & 0.49 (0.031) & 0.032 (0.010) \\
\hline
Turnover=60\%&0.94 (0.007) & 0.005 ($<$0.001) & 1.98 (0.010) & 0.011 ($<$0.001) & 0.27 (0.036) & 0.039 (0.008) & 0.48 (0.027) & 0.031 (0.009) \\
\hline
Turnover=70\%&0.94 (0.007) & 0.005 ($<$0.001) & 1.98 (0.014) & 0.015 ($<$0.001) & 0.26 (0.030) & 0.043 (0.008) & 0.48 (0.025) & 0.033 (0.011) \\
\hline
Turnover=80\%&0.95 (0.005) & 0.005 ($<$0.001) & 1.99 (0.024) & 0.022 ($<$0.001) & 0.26 (0.056) & 0.049 (0.008) & 0.49 (0.038) & 0.034 (0.012) \\
\hline
\end{tabular}}
        \caption{Posterior-based mean (empirical standard deviation) for point estimation and variance estimation for parameters $\alpha$, $\delta$, $\gamma^w$, and $\gamma^b$, based on $n = 500$ nodes with $T = 10$ time points. The table presents the robustness of the estimation method across varying proportions of node turnover, ranging from 0\% to 80\%. Despite increasing turnover, the point estimates remain consistent. The variance of the estimates for $\delta$, $\gamma^w$, and $\gamma^b$ increases due to the reduction in sample size, which is also captured by the variance estimates appropriately.}
    \label{tab:robust}
\end{table}

Table \ref{tab:robust} presents the posterior-based mean (empirical standard deviation) for point estimation and variance estimation, based on a CLSNA model with $n = 500$ nodes, similar to the size of a US Congress, across varying proportions of nodes entering and leaving the network to demonstrate the robustness of the method. Similar to Congress, at each time step, a random subset of nodes are dropped out, while an equal number of new nodes are randomly initialized according to the model's prior distribution. The results indicate that both the point estimates and variance estimates are robust across different turnover levels. As the turnover proportion increases, the variance of the estimates for $\delta$, $\gamma^w$, and $\gamma^b$ also increases, which can be attributed to the smaller sample sizes available for estimating these parameters. This shows the method's capability to provide robust point estimates with varying proportions of nodes entering and leaving while effectively capturing the increase in uncertainty as the turnover proportion rises.
\end{document}